\documentclass[superscriptaddress,aps,pra,twocolumn,showpacs,nofootinbib]{revtex4-1}
\usepackage{etex}
\usepackage{amsmath,amssymb,amsthm}
\usepackage{easybmat}
\usepackage[colorlinks=true,citecolor=blue,urlcolor=blue]{hyperref}
\usepackage[pdftex]{graphicx}
\usepackage{times,txfonts}
\usepackage{braket}
\usepackage{color}
\usepackage{natbib}

\newcommand{\be}{\begin{equation}}
\newcommand{\ee}{\end{equation}}
\newcommand{\ba}{\begin{eqnarray}}
\newcommand{\ea}{\end{eqnarray}}
\newcommand{\ketbra}[2]{|#1\rangle \langle #2|}
\newcommand{\tr}{\operatorname{Tr}}

\newcommand{\etal}{{\it{et. al. }}}

\newtheorem{definition}{Definition}

\newtheorem{example}{Example}
\newtheorem{thm}{Theorem}
\newtheorem{remark}{Remark}
\begin{document}
\title{Detecting genuine multipartite entanglement in steering scenarios}   
   \author{C. Jebaratnam}
   \email{jebarathinam@iisermohali.ac.in;jebarathinam@gmail.com}
\affiliation{Indian Institute of Science Education and Research Mohali, 
Sector-81, S.A.S. Nagar, Manauli 140306, India.}
\affiliation{Department of Physics, Indian Institute of Technology Madras, 
Chennai 600036, India}
\date{\today}
\begin{abstract}
Einstein-Podolsky-Rosen (EPR) steering is a form of quantum nonlocality which is intermediate between entanglement and Bell nonlocality.
EPR steering is a resource for quantum key distribution that is device independent on only one side in that it certifies bipartite entanglement when 
one party's device is not characterized while the other party's device is fully characterized.  
In this work, we introduce two types of genuine tripartite   EPR-steering,    and 
derive two steering inequalities to detect them. In a semi-device-independent scenario where only the dimensions of 
two parties are assumed,
the correlations which violate one of these inequalities also certify genuine tripartite 
entanglement. 
It is known that Alice can demonstrate bipartite EPR-steering to 
Bob if and only if her measurement settings are incompatible. 
%We note that when bipartite quantum correlations exhibit EPR-steering from Alice to Bob, it is necessary that 
%measurement settings of Bob are incompatible.
We demonstrate that quantum correlations can also detect tripartite EPR-steering from Alice to Bob and Charlie,
even if Charlie's measurement settings are compatible.
\end{abstract}
\maketitle
\section{Introduction}
Entanglement, Einstein-Podolsky-Rosen (EPR) steering, and Bell nonlocality are three inequivalent forms of nonlocality in quantum physics 
\cite{WJD,enteprnl}. The observation of Bell nonlocality \cite{bell64, BNL} implies the presence of entanglement without the need for the 
detailed characterization of the measured systems as well as the measurement operators. For this reason, Bell nonlocality  has been used 
as a resource for device-independent (DI) quantum information processing \cite{DQKD,Pironioetal}. EPR-steering is a  weaker  form     of 
quantum nonlocality \cite{Schrodinger,WJD} and is witnessed by violation of a steering inequality \cite{CJWR,EPRsi,CFFW}. 
EPR-steering is a resource for one-side-device-independent ($1$SDI) quantum key distribution \cite{SDIQKD}.
This follows from the fact that the observation of EPR-steering in bipartite systems certifies entanglement with measurements on 
one side characterized and the other side uncharacterized.  

The observation of Bell nonlocality or EPR-steering also implies the presence of another nonclassical feature, which is incompatibility of 
measurements \cite{BNL,IncomN}. %In the case of projective measurements, commutativity characterizes well the notion of incompatibility. 
The notion of commutativity does not properly capture the incompatibility of generalized measurements, i.e., positive-operator-valued measurements (POVMs). 
The notion of joint measurability \cite{nJM}, which is inequivalent to commutativity, is a natural choice to capture the incompatibility of POVMs. 
Recently, it has been shown that a set of POVMs can be used to demonstrate bipartite EPR-steering 
if and only if (iff) it is nonjointly measurable \cite{IncomN, JMguhne,Uola}. In other words,
Alice's measurement settings are incompatible if she
can demonstrate EPR-steering to Bob. On the other hand,
Alice can ``always find a quantum state'' to demonstrate  
EPR-steering to Bob if her measurement settings are incompatible. 

In the multipartite scenario, various approaches to DI verification of genuine entanglement are known 
\cite{SI,multiSI,multiSI1,Banceletal1, DImulti,DImulti1,DImulti2,selftest}. Violation of the Bell-type inequalities which
detect genuine nonlocality \cite{SI,multiSI,multiSI1,Banceletal1} belongs to one of these approaches. 
Verification of multipartite entanglement in partially DI scenarios, where some of the parties' measurements are trusted,
has also been characterized \cite{SDMulti,Gsteer,SDMulti1,SDMulti2}.
In Ref. \cite{Gsteer}, genuine multipartite 
forms of EPR-steering have been developed and criteria to detect them have been derived. The violation of these criteria ensures that steering 
is shared among all subsystems of the multipartite system.
The observation of genuine multipartite steering implies the presence of genuine 
multipartite entanglement in a $1$SDI way.
In Ref. \cite{SDMulti2}, tripartite steering inequalities have been derived to detect genuine tripartite entanglement in $1$SDI scenarios
where either one or two parties' measurements are uncharacterized.

In this paper, we consider two types of tripartite steering scenarios where two parties' measurements 
are fully characterized (see Fig. \ref{plotine}).
We derive two inequalities to detect genuine tripartite steering in these two $1$SDI scenarios.
We argue that the correlations which violate one of these steering 
inequalities detect genuine entanglement in a semi-DI way \cite{SDI} as well.
That is, when the Hilbert-space dimensions of two subsystems are constrained, the correlations which exhibit genuine
tripartite steering also detect genuine entanglement. %(without any assumption on the measurements performed by all three parties). 
%Since a set of POVMs which is jointly measurable cannot be used to demonstrate bipartite steering, 
%it has been considered as classical \cite{JMguhne}. 
We demonstrate that quantum correlations can also detect tripartite EPR-steering from Alice to Bob and Charlie, even if Charlie's measurement settings are compatible.

%This implies that unlike the bipartite case, incompatibility of generalized measurements is not necessary for demonstrating 
%tripartite steering.  

%In the present work we derive two tripartite steering inequalities which detect genuine entanglement 
%in $1$SDI scenarios where two parties' measurements 
%are characterized. %We derive two steering inequalities which detect them. 
%by using the structure of the Svetlichny inequality which detects genuine nonlocality \cite{SI} and 
%Mermin inequality which detects standard nonlocality \cite{mermin}, respectively.
%We argue that the correlations which violate one of these 
%inequalities detect genuine entanglement in semi-DI way as well.
%That is, when the dimensions of two subsystems are constrained, the correlations which exhibit
%tripartite steering detect genuine entanglement without any assumption on measurements performed by all parties. 
%Since a set of POVMs which is jointly measurable cannot be used to demonstrate bipartite steering, 
%it has been considered as classical \cite{JMguhne}. 
%Unlike the bipartite case, we demonstrate that the compatible general measurements can be used to demonstrate 
%tripartite steering. That is, Charlie can demonstrate tripartite steering to Alice and
%Bob even if Charlie's measurement settings are compatible. 

The paper is organized as follows. In Sec. \ref{biNL}, we discuss in detail the three notions of bipartite nonlocality. 
In Sec. \ref{Gsteer}, we introduce two notions of genuine tripartite EPR-steering and we derive the steering inequalities which detect them. 
In Sec. \ref{IncVsEn}, we illustrate that incompatibility of POVMs  is not necessary for detecting tripartite steering.
Conclusions are presented in Sec. \ref{Cnc}.

\begin{figure}
\centering
\includegraphics[width=0.40\textwidth]{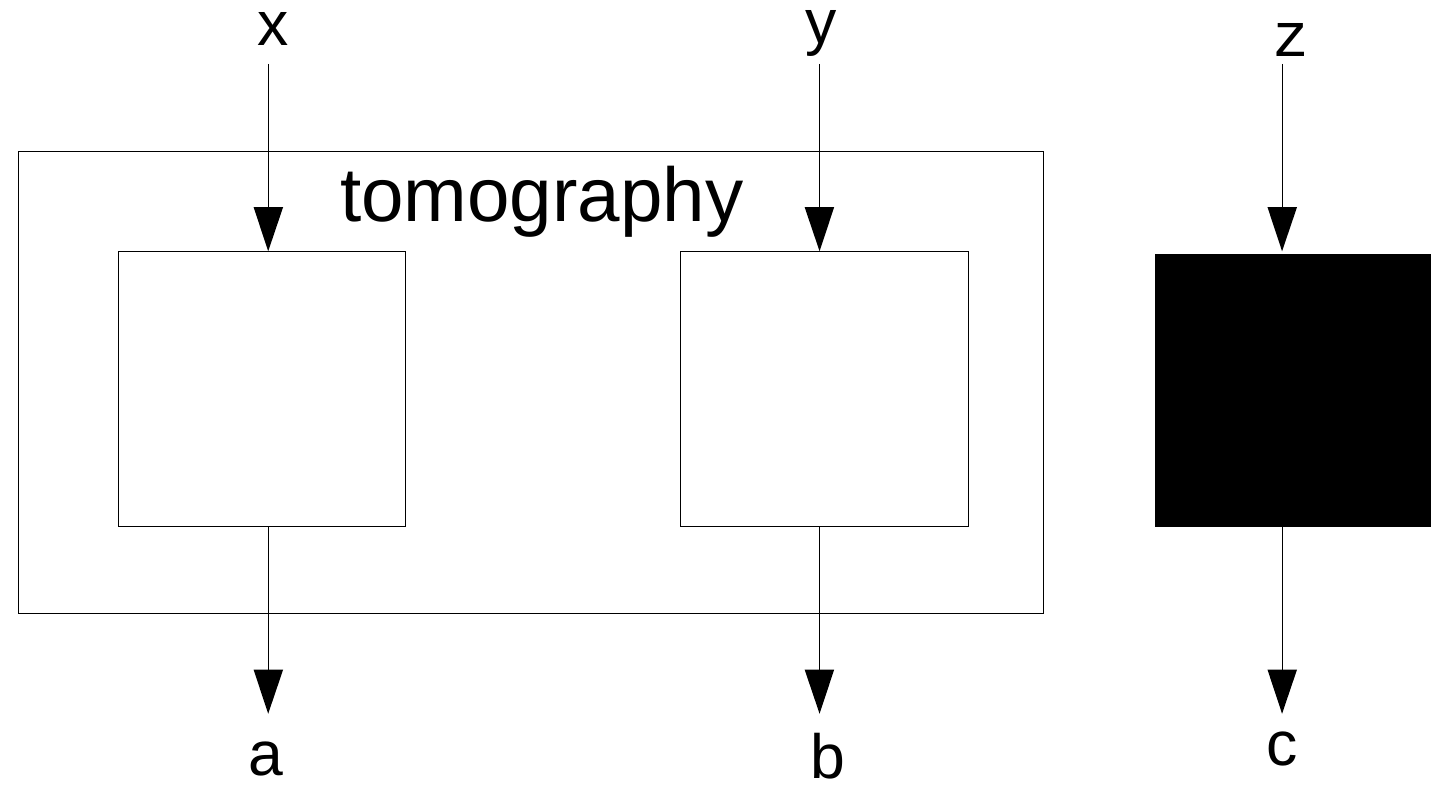} 
\caption{\textit{Depiction of the one-side-device-independent scenario that we consider for deriving the two steering inequalities:} 
Alice and Bob's subsystems and measurements are fully characterized.
Charlie's measurements are treated as black box, i.e., Charlie has no knowledge about the measurement operators as well as the measured system.
}\label{plotine}
\end{figure}
\section{Bipartite nonlocality}\label{biNL}
Consider a bipartite scenario in which two spatially separated parties, Alice and Bob, perform local measurements on a composite quantum system 
described by a density operator $\rho\in \mathcal{B}(\mathcal{H}_A\otimes \mathcal{H}_B)$. The correlation between the outcomes is described 
by the conditional probability distributions of getting the outcomes $a$ and $b$ given the measurements $A_x$ and $B_y$ on Alice and 
Bob's sides: $P(ab|A_xB_y)$, where $x$ and $y$ label the measurement choices. Quantum theory predicts the correlation as follows:
\be
P(ab|A_xB_y)=\tr (\rho M_{a|x} \otimes M_{b|y}),
\ee
where $M_{a|x}$ and $M_{b|y}$ are measurement operators.

\textit{Bell nonlocality.} A correlation is Bell nonlocal if it cannot 
be explained by the local hidden variable (LHV) model \cite{BNL},
\be
P(ab|A_xB_y)=\sum_\lambda p_\lambda P_\lambda(a|A_x)P_\lambda(b|B_y),
\ee
for some hidden variable $\lambda$ with probability distribution $p_\lambda$.
In the case of two dichotomic measurements per side, i.e., $x,y\in\{0,1\}$ and $a,b\in\{-1,+1\}$, the  
correlation has a LHV model iff it satisfies the Bell--Clauser-Horne-Shimony-Holt (CHSH) inequality \cite{chsh},
\be
\braket{A_0B_0+A_0B_1+A_1B_0-A_1B_1}_{LHV}\le2, \label{BCHSH}
\ee
and its equivalents \cite{Fine}. Here, $\braket{A_xB_y}=\sum_{ab}abP(ab|A_xB_y)$. Quantum correlations violate the Bell-CHSH 
inequality up to the Tsirelson bound $2\sqrt{2}$ \cite{tsi1}.

\textit{EPR-steering.}
Consider a $1$SDI scenario in which Alice has knowledge about her subsystem and which measurements she can perform,
while Bob performs black-box measurements (i.e., uncharacterized measurements). In this scenario, a
quantum correlation exhibits EPR-steering from Bob to Alice if it cannot be explained by the hybrid local 
hidden state (LHS)-LHV model \cite{CFFW},
\be
\tr (\rho M_{a|x} \otimes M_{b|y})=\sum_\lambda p_\lambda P(a|A_x,\rho_\lambda)P_\lambda(b|B_y),
\ee
where $P(a|A_x,\rho_\lambda)$ is the distribution arising from quantum state $\rho_\lambda$. Suppose Alice performs two orthogonal qubit 
projective measurements; for instance, $\sigma_x$ and $\sigma_y$, and Bob performs two dichotomic black-box measurements.  
Then the inequality
\be
\braket{A_0B_0-A_1B_1}_{2\times ?}^{LHS}\le\sqrt{2}, \label{EPRB}
\ee
where $A_0$ and $A_1$ are qubit projective measurements which satisfy $[A_0,A_1]=-1$, serves as the EPR-steering criterion \cite{CJWR,EPRsi}. 
Here, $2 \times ?$ indicates that Alice's subsystem is assumed to be qubit while Bob's subsystem is uncharacterized. Just like the Bell-CHSH 
inequalities, there are eight equivalent steering inequalities. For instance, the steering inequality 
$\braket{A_0B_1+A_1B_0}_{2\times?}^{LHS}\le\sqrt{2}$ can be obtained from the one in Eq. (\ref{EPRB}) by the transformation 
$a \rightarrow a\oplus x$, with $a\in\{0,1\}$ and $y\rightarrow y\oplus1$; here, $\oplus$ denotes addition modulo $2$.

Suppose Bob performs unknown measurements with POVM elements $M_{b|y}$ on his share of a bipartite quantum state 
$\rho^{AB}\in\mathcal{B}(\mathcal{H}_2\otimes \mathcal{H}_d)$.
The unnormalized conditional states on Alice's side are given by 
$\sigma^A_{b|y}=\tr(\openone \otimes M_{b|y}\rho^{AB})$.  
Alice can do state tomography to determine these conditional states.
The set of unnormalized conditional states is called an assemblage \cite{EPRLHS}.
The above scenario exhibits steering if the state assemblage does not have a LHS model,
\be
\sigma^A_{b|y}=\sum_\lambda p_\lambda P_\lambda (b|y) \rho_\lambda,   \label{SABA}
\ee
where $P_\lambda (b|y)$ are some conditional probability distributions 
and $\rho_\lambda$ are positive operators which satisfy 
$\sum_\lambda \tr  p_\lambda  \rho_\lambda=1$. %Here $\lambda$ is a hidden variable.
The violation of a steering inequality as in Eq. (\ref{EPRB}) implies that the assemblage does not have a
decomposition as in Eq. (\ref{SABA}) \cite{EPRLHS}. Note that the state assemblage arising from any separable state has a LHS model. This implies that when the
$1$SDI scenario does not demonstrate steering, there always exists a separable state which reproduces the given state assemblage \cite{Piani}.

\textit{Nonseparability.} Nonseparability of quantum correlations arises as a failure of the quantum separable model. In this model, 
LHS description is assumed for both of the parties. When a quantum correlation exhibits nonseparability, 
it violates a LHS-LHS model,
\be
\tr (\rho M_{a|x} \otimes M_{b|y})=\sum_\lambda p_\lambda P(a|A_x,\rho^A_\lambda)P(b|B_y,\rho^B_\lambda),
\ee
where $P(a|A_x,\rho^A_\lambda)$ and $P(b|B_y,\rho^B_\lambda)$ are the distributions arising from quantum states $\rho^A_\lambda$ 
and $\rho^B_\lambda$, respectively. Any condition that is derived under the assumption of the above model is known as separability criterion
or entanglement criterion.

\subsection{Entanglement certification from CHSH and BB84 families}
Moroder and Gittsovich (MG) \cite{DbBi} explored the task of entanglement detection in various partially DI scenarios. 
For instance, MG introduced entanglement certification in the semi-DI scenario where only the Hilbert-space dimensions of the subsystems are assumed.
In Ref. \cite{Koon}, Goh \etal have defined a quantity  which gives certifiable lower bounds on the amount of entanglement  
in the semi-DI scenario. 
By using this quantity, Goh \etal studied the amount of two-qubit entanglement certifiable from the CHSH family defined as
\be
P_{CHSH}=\frac{2+ab(-1)^{xy}\sqrt{2}V}{8}, \label{chshfam}
\ee
and the BB84 family  defined as
\be
P_{BB84}=\frac{1+ab\delta_{x,y}V}{4}. \label{bb84fam}
\ee
The CHSH family with $V=1$ violates the Bell-CHSH inequality to its quantum bound of $2\sqrt{2}$, whereas the BB84 family  
with $V=1$ corresponds to the BB84 correlations \cite{DQKD}. 
The CHSH and BB84 families can be obtained from
the two-qubit Werner state, $\rho_W=V\ketbra{\Psi^-}{\Psi^-}+(1-V)\openone/4$,
where $\ket{\Psi^-}=(\ket{01}-\ket{10})/\sqrt{2}$,
for suitable projective measurements. 
The CHSH family is achievable with the measurement settings that give rise to the maximal violation of the CHSH inequality in Eq. (\ref{BCHSH}), whereas
the BB84 family is achievable with the settings that give rise to the maximal violation of the steering inequality in Eq. (\ref{EPRB}). 
%We will now analyze how much
%entanglement of the Werner state is certifiable from these two families. 

As the CHSH family violates the Bell-CHSH inequality for $V>1/\sqrt{2}$, it certifies entanglement in a DI way in this range.
Since the BB84 family is local, it can also be produced by a separable state in the higher-dimensional space \cite{DQKD}. 
However, entanglement is certifiable from the BB84 family for $V>1/\sqrt{2}$ in a $1$SDI way, 
as it violates the steering inequality in this range.  

Note that the two-qubit Werner state is entangled iff $V>1/3$ \cite{Werner}.
If one assumes qubit dimension and which measurements are performed on both sides, the CHSH and BB84 families detect entanglement for $V>1/2$. 
This follows from the fact that these correlations violate a quantum separable model in this range.
This can be checked by the criteria to detect the nonexistence of a LHS-LHS model derived in Ref. \cite{UFNL}.
The CHSH family with $V>1/2$ violates the following LHS-LHS condition:
\be
\braket{A_0B_0+A_0B_1+A_1B_0-A_1B_1}_{2\times2}^{LHS}\le\sqrt{2},
\ee
with the assumption that Alice and Bob have access to measurements that give rise to the optimal 
violation of the Bell-CHSH inequality in Eq. (\ref{BCHSH}).
The BB84 family with $V>1/2$ violates the following LHS-LHS condition:
\be
\braket{A_0B_0-A_1B_1}_{2\times2}^{LHS}\le1, \label{Nonsep}
\ee
with the assumption that Alice and Bob have access to measurements that give rise to the optimal violation 
of the steering inequality in Eq. (\ref{EPRB}).
Goh \etal found that for $V>1/2$,
entanglement is certifiable from these two families if one assumes only qubit systems for Alice and Bob \cite{Koon}. 
%Note that criterion given by MG in Ref. \cite{DbBi} also shows that in this range, both CHSH and BB84 families detect entanglement
%when .  
\subsection{Incompatibility vs nonseparability}
We now note that incompatibility of measurements is not necessary to demonstrate nonseparability. 
For this, we consider a measurement scenario in which Alice has access to the set of two dichotomic qubit POVMs, 
$\mathcal{M}^\eta_A=\{M^\eta_{\pm|\hat{a}_x}|x=0,1\}$, with
elements 
\be
M^\eta_{\pm|\hat{a}_x}=\eta \Pi_{\pm|\hat{a}_x}+(1-\eta)\frac{\openone}{2}; \quad 0 \le \eta \le1, \label{nPVM}
\ee
and Bob has access to the set of two projective qubit measurements $\mathcal{M}_B=\{\Pi_{\pm|\hat{b}_y}|y=0,1\}$. 
Here $\Pi_{\pm|\hat{a}_x}=1/2(\openone \pm \hat{a}_x \cdot \vec{\sigma})$ and $\Pi_{\pm|\hat{b}_y}=1/2(\openone \pm \hat{b}_y \cdot \vec{\sigma})$
are projectors along the directions $\hat{a}_x$ and $\hat{b}_y$; $\vec{\sigma}$ is the vector of Pauli matrices.
Note that in the above measurement scenario,
Alice performs the noisy projective measurements. This implies that the correlation arising from the given two-qubit state $\rho$ 
satisfies the following relation:
\be
\tr (\rho M^\eta_{\pm|\hat{a}_x} \otimes \Pi_{\pm|\hat{b}_y})=\tr(\rho_\eta \Pi_{\pm|\hat{a}_x} \otimes \Pi_{\pm|\hat{b}_y}). \label{povm-proj}
\ee
Here, $\rho_\eta=\eta \rho+ (1-\eta)\openone/2 \otimes \rho_B$, with $\rho_B=\tr_A\rho$. In other words, the correlation arising from 
the given state for noisy projective measurements on Alice's side and projective measurements on Bob's side is equivalent to the correlation
arising from the noisy state for the projective measurements on both sides. In Refs. \cite{JMguhne,IncompLHV}, this connection between the 
noisy measurements and noisy states has been used to obtain new results for the former from known results of the latter. 
%Similarly, this connection leads us to observe that 
%incompatibility is not necessary for demonstrating tripartite EPR-steering 

The set of two POVMs $\mathcal{M}^\eta_A$ with $\hat{a}_0\cdot \hat{a}_1=0$ is jointly measurable iff $\eta\le1/\sqrt{2}$ \cite{JMPauli1,JMPauli2}
and noncommuting for any $\eta>0$.   
We now illustrate that this noncommuting set of POVMs can be used to demonstrate nonseparability for $\eta>1/2$. 
\begin{example}\label{ex01}
Suppose Alice and Bob
share the singlet state, $\ket{\Psi^-}$, with Alice performing two noisy projective measurements along the directions $\hat{a}_0=\hat{x}$ and $\hat{a}_1=\hat{y}$ and
Bob performing two projective measurements along the directions $\hat{b}_0=-(\hat{x}+\hat{y})/\sqrt{2}$ and $\hat{b}_1=(-\hat{x}+\hat{y})/\sqrt{2}$.
From the relation given in Eq. (\ref{povm-proj}), it follows that the statistics arising from this setting are analogous to that arising from the Werner state
with the visibility $\eta$
for projective measurements on both sides along the directions given above.
Therefore, the statistics arising from the above scenario are equivalent to the CHSH family in Eq. (\ref{chshfam}) with $V$ replaced by $\eta$.
\end{example}
Since the above statistics exhibit nonseparability for $\eta>1/2$, 
the set of two POVMs which is jointly measurable in this range is also useful for this nonclassical task. 
This, in turn, implies that this compatible set of POVMs can also be used to detect two-qubit entanglement.
For the measurements given in example \ref{ex01}, the statistics arising from the Werner state are equivalent 
to the CHSH family in Eq. (\ref{chshfam}) with $V$ replaced by $\eta V$. Thus, these statistics detect entanglement for any $\eta V >1/2$. 
For $\eta=1/\sqrt{2}$, the statistics do not exhibit Bell nonlocality; however, entanglement is detected for $V>1/\sqrt{2}$.
\section{Tripartite nonlocality}\label{Gsteer}
We now turn to the tripartite case which is the focus of this work.
We restrict ourselves to the simplest scenario in which three spatially separated parties, i.e., Alice, Bob, and Charlie, 
perform two dichotomic measurements on their subsystems.
The correlation is described by the conditional probability distributions: $P(abc|A_xB_yC_z)$, where $x,y,z\in\{0,1\}$ and $a,b,c\in\{-1,+1\}$. 
The correlation exhibits standard nonlocality (i.e., Bell nonlocality) if it cannot be explained by the
LHV model,
\be
P(abc|A_xB_yC_z)=\sum_\lambda p_\lambda P_\lambda(a|A_x)P_\lambda(b|B_y)P_\lambda(c|C_z). \label{FLHV}
\ee
If a correlation cannot be reproduced by this fully LHV model, it does not necessarily imply that it exhibits genuine nonlocality \cite{SI,Banceletal1}.

In Ref. \cite{SI}, Svetlichny introduced the strongest form of genuine tripartite nonlocality 
(see Ref. \cite{Banceletal1} for the other two forms 
of genuine nonlocality).
A correlation exhibits Svetlichny nonlocality if it cannot be explained by a hybrid nonlocal-LHV (NLHV) model,
\begin{align}
&P(abc|A_xB_yC_z)\!=\!\sum_\lambda p_\lambda P_\lambda(a|A_x)P_\lambda(bc|B_yC_z)+\nonumber \\
&\!\sum_\lambda q_\lambda P_\lambda(ac|A_xC_z)P_\lambda(b|B_y)\!+\!\sum_\lambda r_\lambda P_\lambda(ab|A_xB_y)P_\lambda(c|C_z), \label{HNLHV}
\end{align}
with $\sum_\lambda p_\lambda+\sum_\lambda q_\lambda+\sum_\lambda r_\lambda=1$. The bipartite probability distributions in this decomposition can have arbitrary nonlocality.

Svetlichny derived Bell-type inequalities to detect the strongest form of genuine nonlocality \cite{SI}. For instance, one of the 
Svetlichny inequalities reads,
\ba
&&\braket{A_0B_0C_1+A_0B_1C_0+A_1B_0C_0-A_1B_1C_1}\nonumber \\
&&+\braket{A_0B_1C_1+A_1B_0C_1+A_1B_1C_0-A_0B_0C_0}\le4. \label{SI1}
\ea
Here $\braket{A_xB_yC_z}=\sum_{abc}abcP(abc|A_xB_yC_z)$. 
Quantum correlations violate the Svetlichny inequality (SI) up to $4\sqrt{2}$. A Greenberger-Horne-Zeilinger (GHZ)
state \cite{GHZ} gives rise to the maximal violation of the SI for a different choice of measurements which do not demonstrate the GHZ paradox \cite{UNLH}.

Bancal \etal \cite{Bancaletal} presented an intuitive approach to the SI in Eq. (\ref{SI1}).
For this, the SI was rewritten as follows:
\be
\braket{CHSH_{AB}C_1+CHSH'_{AB}C_0}_{NLHV}\le4. \label{SI}
\ee
Here, $CHSH_{AB}$ is the canonical CHSH operator given in Eq. (\ref{BCHSH}) and $CHSH'_{AB}=-A_0B_0+A_0B_1+A_1B_0+A_1B_1$ is one of its equivalents.
Bancal \etal observed that the input setting of Charlie defines which version of CHSH game Alice 
and Bob are playing. When $C$ gets the input $z=0$, $AB$ play the canonical $CHSH$ game; when $C$ gets the input $z=1$,
$AB$ play $CHSH'$. In Argument $1$ of Ref. \cite{Bancaletal}, Bancal \etal found that $AB$ 
play the average game $\pm CHSH \pm CHSH'$, where the signs indicate that which game they are playing depends on the outputs of $C$. 
It can be checked that the algebraic maximum of any of these average games is $4$ for the NLHV
model in Eq. (\ref{HNLHV}) with $\sum_\lambda r_\lambda=1$.

\subsection{Svetlichny steering}
We will derive a criterion for tripartite EPR-steering by exploiting the structure of the SI given in Eq. (\ref{SI}).
For this, we consider the following $1$SDI scenario. 
%We will now define one of the asymmetric forms of Svetlichny steering.
Alice and Bob have access to incompatible qubit measurements that give rise to violation of a Bell-CHSH inequality, while Charlie
has access to two black-box measurements. 
Suppose that $P(abc|A_xB_yC_z)$ cannot be explained by the following nonlocal LHS-LHV (NLHS)
model:
\begin{align}
P(abc|A_xB_yC_z)&=\sum_\lambda p_\lambda P(ab|A_xB_y,\rho_{AB}^\lambda)P_\lambda(c|C_z)+\nonumber \\
&\sum_\lambda q_\lambda P(a|A_x,\rho_A^\lambda)P(b|B_y,\rho_B^\lambda)P_\lambda(c|C_z), \label{NLHS}
\end{align}
where $P(ab|A_xB_y,\rho_\lambda)$ denotes the nonlocal probability distribution arising from two-qubit state $\rho_{AB}^\lambda$, and $P(a|A_x,\rho_A^\lambda)$ 
and $P(b|B_y,\rho_B^\lambda)$ are the distributions arising from qubit states $\rho_A^\lambda$ and $\rho_B^\lambda$. 
Then the quantum correlation exhibits genuine steering from Charlie to Alice and Bob.

We obtain the following criterion for genuine steering under the constraint of the above $1$SDI scenario.
\begin{thm}
If a given quantum correlation violates  
the steering inequality
\be
\braket{CHSH_{AB}C_1+CHSH'_{AB}C_0}_{2\times 2 \times ?}^{NLHS}\le2\sqrt{2}, \label{SIEPR}
\ee
then it exhibits genuine tripartite steering from Charlie to Alice and Bob.
Here, $2\times 2 \times ?$
indicates that Alice and Bob have access to known qubit measurements that demonstrate Bell nonlocality, while Charlie's measurements are uncharacterized. 
\end{thm}
\begin{proof}
Note that in the $1$SDI scenario that we are interested in, $AB$ play the average game $\pm CHSH \pm CHSH'$. 
The maximum of any of these games cannot exceed $2\sqrt{2}$ if the correlation admits the NLHS model given in Eq. (\ref{NLHS}).
There are two cases which have to be checked:
(i) Suppose Alice and Bob have a LHS-LHS model, the expectation values of the CHSH operators are bounded by 
$-\sqrt{2}\le\braket{CHSH}^{LHS}_{2\times2}\le \sqrt{2}$ and  $-\sqrt{2}\le\braket{CHSH'}^{LHS}_{2\times2}\le \sqrt{2}$ \cite{UFNL}.
This implies that $\braket{\pm CHSH \pm CHSH'}^{LHS}_{2\times2}\le 2\sqrt{2}$. 
(ii) In case Bell nonlocality is shared by Alice and Bob, it can be checked that  
$\braket{\pm CHSH \pm CHSH'}_{2\times 2}\le 2\sqrt{2}$. For instance, the quantum correlation which exhibits maximal Bell nonlocality 
has $\braket{CHSH}=2\sqrt{2}$ and $\braket{CHSH'}=0$. Thus, the violation of the inequality in Eq. (\ref{SIEPR}) implies
the violation of the model in Eq. (\ref{NLHS}).
\end{proof}

For a given quantum state, genuine steering as witnessed by the steering inequality in Eq. (\ref{SIEPR}) originates from
measurement settings that give rise to Svetlichny nonlocality. 
For this reason, we call this type of genuine steering  Svetlichny steering.
\begin{definition}
A quantum correlation exhibits \emph{Svetlichny steering} if it cannot be explained by a model in which arbitrary
two-qubit Bell nonlocality is allowed between two parties with the third party locally correlated.
\end{definition}
Note that the SI is invariant under the permutation of the parties. This implies that the criterion to detect 
Svetlichny steering from $A$ to $BC$ or $B$ to $AC$ can be obtained from the steering inequality in Eq. (\ref{SIEPR}) by permuting the parties.

We consider the Svetlichny family defined as
\be
P_{Sv}=\frac{2+abc(-1)^{xy\oplus xz \oplus yz \oplus x \oplus y \oplus z \oplus 1}\sqrt{2}V}{16}, \label{SvF}
\ee
which is the tripartite version of the CHSH family in Eq. (\ref{chshfam}).
\begin{example}\label{ex1}
The Svetlichny family can be obtained from a noisy three-qubit GHZ state,
$\rho=V\ketbra{\Phi_{GHZ}}{\Phi_{GHZ}}+(1-V)\openone/8$, where $\ket{\Phi_{GHZ}}=\frac{1}{\sqrt{2}}(\ket{000}+\ket{111})$, for the measurements that 
give rise to the maximal violation of the SI; for instance,
$A_0=\sigma_x$, $A_1=\sigma_y$, $B_0=(\sigma_x-\sigma_y)/\sqrt{2}$, $B_1=(\sigma_x+\sigma_y)/\sqrt{2}$, 
$C_0=\sigma_x$, and $C_1=-\sigma_y$. 
\end{example}
Note that the noisy GHZ state given above is genuinely entangled iff $V>0.429$ \cite{Guhne}.
The Svetlichny family certifies genuine entanglement in a DI way for $V>1/\sqrt{2}$,
as it violates the SI in this range. 

The Svetlichny family can be written as a convex mixture of the local deterministic strategies when $V\le 1/\sqrt{2}$ \cite{Banceletal1}. This implies that
in this range, it can also arise from a separable state in the higher dimensional space \cite{DQKD}. 
However, the measurement statistics in example \ref{ex1} certify genuine entanglement for 
$V>1/2$ in a $1$SDI way, as these statistics violate the 
steering inequality in Eq. (\ref{SIEPR}).

By using the concept of steering, Bancal \etal \cite{Bancaletal} observed that the structure of the SI in Eq. (\ref{SI}) allows one to understand
its violation by genuinely entangled states: 
The SI should be violated iff Charlie's measurements prepare entangled states for Alice and Bob 
such that the average game $\pm CHSH \pm CHSH'>4$. 
Similarly, we understand violation of the steering inequality by the Svetlichny family as follows.
When the parties share the noisy GHZ state and observe the optimal violation of the steering inequality in Eq. (\ref{SIEPR}), we have
the following two situations.
First, Charlie's black-box measurements prepare the noisy Bell states, which are a mixture of a Bell state (i.e., a maximally 
entangled state) and white noise, with the visibility $V$
for Alice and Bob.
Second, Alice and Bob's measurements generate the CHSH family from the states steered by Charlie. 
These imply that the Svetlichny family certifies genuine tripartite entanglement for $V>1/2$ if one assumes only 
qubit dimension for two parties.

\subsection{Mermin steering}
We now derive a criterion for another form of tripartite EPR-steering from a Mermin inequality (MI) \cite{mermin},
\be
\braket{A_0B_0C_1+A_0B_1C_0+A_1B_0C_0-A_1B_1C_1}_{LHV}\le2. \label{MI0}
\ee
In the seminal paper \cite{mermin}, the MI was derived to demonstrate standard nonlocality of 
three-qubit correlations arising from the genuinely entangled states. 
For this purpose, noncommuting measurements
that do not demonstrate Svetlichny nonlocality were used.
Note that when the GHZ state maximally violates the MI, the measurements
that give rise to it exhibit the GHZ paradox \cite{UNLH}.

We rewrite the MI in Eq. (\ref{MI0}) as follows:
\be
\braket{Mermin_{AB}C_1+Mermin'_{AB}C_0}_{LHV}\le2, \label{MI}
\ee
where $Mermin_{AB}=A_0B_0-A_1B_1$ and $Mermin'_{AB}=A_0B_1+A_1B_0$ \footnote{The multipartite generalization of these 
operators generates the Mermin inequalities, hence the name \cite{mermin,allmult}.}. 
Note that these bipartite Mermin operators 
can be used to witness EPR-steering without Bell nonlocality as in Eq. (\ref{EPRB}).
Inspired by the structure of the MI in Eq. (\ref{MI}), we consider the following 
$1$SDI scenario. 
Alice and Bob have access to incompatible qubit measurements that give rise to 
EPR-steering without Bell nonlocality, while Charlie has access to two dichotomic black-box measurements.
Suppose that $P(abc|A_xB_yC_z)$ cannot be explained by the following steering LHS-LHV (SLHS) model:
\begin{align}
P(abc|A_xB_yC_z)&=\sum_\lambda p_\lambda P^{EPR}_L(ab|A_xB_y,\rho_{AB}^\lambda)P_\lambda(c|C_z)+\nonumber \\
&\sum_\lambda q_\lambda P(a|A_x,\rho_A^\lambda)P(b|B_y,\rho_B^\lambda)P_\lambda(c|C_z),\label{wSVEPR}
\end{align}
where $P^{EPR}_L(ab|A_xB_y,\rho_{AB}^\lambda)$ denotes the local probability distribution which exhibits EPR-steering.
Then the quantum correlation exhibits genuine steering from Charlie to Alice and Bob.

We obtain the following criterion which witnesses genuine steering in the above $1$SDI scenario.
\begin{thm}
If a given quantum correlation violates the steering inequality
\be
\braket{Mermin_{AB}C_1+Mermin'_{AB}C_0}_{2\times2\times?}^{SLHS}\le2, \label{MIEPR}
\ee
then it exhibits genuine tripartite steering from Charlie to Alice and Bob. Here Alice and Bob's 
measurements demonstrate EPR-steering without Bell nonlocality, while Charlie's measurements
are uncharacterized. 
\end{thm}
\begin{proof}
Notice that the correlations that admit the SLHS model in Eq. (\ref{wSVEPR})
also admit the standard LHV model in Eq. (\ref{FLHV}). This implies that if a given correlation violates the 
inequality in Eq. (\ref{MIEPR}), the correlation also violates the SLHS model.

Similar to the proof of the inequality in Eq. (\ref{SIEPR}), we will also show that the inequality in Eq. (\ref{MIEPR})
is satisfied if the correlation admits the SLHS model.
Notice that in the $1$SDI scenario that we are now interested in, $AB$ play the average game 
$\pm Mermin \pm Mermin'$. The algebraic maximum of any of these games is $2$ if the correlation admits 
the SLHS model. There are two cases that have to be checked now.
(i) Suppose Alice and Bob have a LHS-LHS model, the expectation values of the Mermin operators satisfy
$-1 \le \braket{Mermin}^{LHS}_{2\times2}\le 1$ and  $-1 \le \braket{Mermin'}^{LHS}_{2\times2}\le 1$ \cite{UFNL}.
This implies that $\braket{\pm Mermin \pm Mermin'}_{2\times2}^{LHS}\le2$.
(ii) In case EPR-steering is shared by Alice and Bob, it can be checked that  
$\braket{\pm Mermin \pm Mermin'}\le 2$. For instance, the quantum correlation which violates the steering inequality in Eq. (\ref{EPRB})
maximally 
has $\braket{Mermin}=2$ and $\braket{Mermin'}=0$. 
\end{proof}

For a given state, genuine steering as witnessed by the steering inequality in Eq. (\ref{MIEPR}) originates from measurement settings
that lead to standard nonlocality, which is witnessed only by the violation of the MI. For this reason, we call this type of tripartite steering 
 Mermin steering.
\begin{definition}
A quantum correlation exhibits \emph{Mermin steering} if it cannot be explained by a model in which arbitrary 
two-qubit EPR-steering without Bell nonlocality is allowed between two parties with the third party locally correlated.  
\end{definition}
Note that the MI is invariant under the permutation of the parties. This implies that the criterion to detect 
Mermin steering from $A$ to $BC$ or $B$ to $AC$ can be obtained from the steering inequality in Eq. (\ref{MIEPR}) by permuting the parties.

We consider the GHZ family defined as
\be
P_{GHZ}=\frac{1+abc\delta_{x\oplus y,z\oplus1}V}{8},
\ee
which is the tripartite version of the BB84 family in Eq. (\ref{bb84fam}).
\begin{example}\label{ex2}
The GHZ family can be obtained from the noisy three-qubit GHZ state for the measurements 
that give rise to the GHZ paradox; for instance, $A_0=\sigma_x$, $A_1=\sigma_y$, $B_0=\sigma_x$, $B_1=\sigma_y$,
$C_0=\sigma_x$, and $C_1=-\sigma_y$.
\end{example}
Since the GHZ family does not exhibit genuine nonlocality,
it does not certify genuine entanglement in a DI way. However, genuine entanglement is detected from the measurement statistics in example \ref{ex2} for $V>1/2$
in a $1$SDI way, as these statistics violate the steering inequality in Eq. (\ref{MIEPR}). 
Notice that when the GHZ family violates this steering inequality, Alice and Bob's measurement settings
generate the BB84 family from the noisy Bell states
with the visibility $V$ which are steered by Charlie.
%Notice that measurements of Alice and Bob generate the BB84 family from the noisy Bell states with the visibility $V$ which are steered by Charlie. 
This implies that genuine entanglement is certifiable
from the GHZ family for $V>1/2$ by assuming only qubit systems for two parties. 
%Further, genuine entanglement is detected from the GHZ family for $V>1/\sqrt{2}$ 
%by assuming only orthogonal qubit projective measurements for any one of the parties, as the BB84 family violates a steering inequality in this range.

\section{Incompatibility versus tripartite steering}\label{IncVsEn}
Following the approach of Ref. \cite{SDMulti2}, we will now discuss tripartite steerability in the $1$SDI scenarios (see Fig. \ref{plotine}) which we have considered. 
Let $\rho^{ABC}\in \mathcal{B}(\mathcal{H}_2\otimes \mathcal{H}_2 \otimes \mathcal{H}_d)$ denote the shared tripartite quantum state and $M_{c|z}$ 
denote the measurement operators on Charlie's side. The assemblage, i.e., the set of (unnormalized) conditional states on Alice and Bob's side, is given by
\be
\sigma^{AB}_{c|z}=\tr_C (\openone \otimes \openone \otimes M_{c|z} \rho^{ABC}).
\ee

Note that in examples \ref{ex1} and \ref{ex2}, there are genuinely tripartite entangled states which do not demonstrate
tripartite steering.
Suppose the $1$SDI scenario which uses genuine tripartite entanglement does not demonstrate tripartite steering. 
The biseparable state which can reproduce the given assemblage can be written as
\ba
\rho_{bisep}^{ABC}&=&\sum_\lambda p^{A:BC}_\lambda \rho^{A}_\lambda \otimes \rho^{BC}_\lambda +\sum_\mu p^{B:AC}_\mu \rho^{B}_\mu \otimes \rho^{AC}_\mu \nonumber \\
&+&\sum_\nu p^{AB:C}_\nu \rho^{AB}_\nu \otimes \rho^{C}_\nu.
\ea
Thus, the unsteerable assemblage admits the following LHS model:
\ba
\sigma^{AB}_{c|z}&=&\sum_\lambda p^{A:BC}_\lambda \rho^{A}_\lambda \otimes \sigma^{B}_{c|z,\lambda} +\sum_\mu p^{B:AC}_\mu \sigma^{A}_{c|z,\mu} \otimes \rho^{B}_\mu \nonumber  \\
&+&\sum_\nu p^{AB:C}_\nu P_\nu(c|z) \rho^{AB}_\nu, \label{UStri}
\ea
where $\sigma^{B}_{c|z,\lambda}=\tr_C(\openone \otimes M'_{c|z}\rho^{BC}_\lambda)$, $\sigma^{A}_{c|z,\mu}=\tr_C(\openone \otimes M'_{c|z}\rho^{AC}_\mu)$ 
and $P_\nu(c|z)=\tr (M'_{c|z} \rho^{C}_\nu)$. %The above form is a LHS model for the state assemblage.

Suppose Charlie has access to POVMs which are compatible. 
Then there exists a parent POVM with measurement operators $G_\nu$ such that
\be
M_{c|z}=\sum_{\nu}D_{\nu}(c|z)G_\nu,
\ee
where $D_{\nu}(c|z)$ are positive numbers with $\sum_{c}D_{\nu}(c|z)=1$ \cite{Ali2009}. For these measurements, 
the state assemblage of Alice and Bob's side admits
the following decomposition:
\be
\sigma^{AB}_{c|z}=\sum_{c}D_{\nu}(c|z) \tr_C (\openone \otimes \openone \otimes G_\nu \rho^{ABC}). \label{compLHS}
\ee
Note that the above decomposition resembles that of the unsteerable assemblage given in Eq. (\ref{UStri}).
Therefore, if Charlie's measurement settings are compatible then there is no steering between Charlie and Alice-Bob.

Inspired by example \ref{ex01}, we consider the following measurement scenario.
\begin{example}\label{ex3}
Alice and Bob perform projective measurements along the directions $\hat{a}_0=\hat{x}$, $\hat{a}_1=\hat{y}$, $\hat{b}_0=(\hat{x}-\hat{y})/\sqrt{2}$ 
and $\hat{b}_1=(\hat{x}+\hat{y})/\sqrt{2}$. Charlie performs noisy projective measurements with visibility $\eta$ as in Eq. (\ref{nPVM}) along the 
directions $\hat{c}_0=\hat{x}$ and $\hat{c}_1=-\hat{y}$. For the above measurements, the statistics arising from the GHZ state, $\ket{\Phi_{GHZ}}$,
are equivalent to the Svetlichny family in Eq. (\ref{SvF}) with $V$ replaced by $\eta$.
\end{example}
Note that in the above example, Charlie's measurement settings are
compatible for $\eta \le 1/\sqrt{2}$. This implies that in this range, the state assemblage of Alice and Bob's side has the decomposition 
as in Eq. (\ref{compLHS}). Therefore, Charlie does not demonstrate tripartite steerability to
Alice and Bob for $\eta\le1/\sqrt{2}$.
However, the correlations in example \ref{ex3} violate the steering inequality in Eq. (\ref{SIEPR}) for any $\eta>1/2$.
Note that in this example, Bob and Charlie's measurement settings can be used to demonstrate Bell nonlocality
for $\eta>1/\sqrt{2}$ and nonseparability for $\eta>1/2$ (as in example \ref{ex01}). 
Therefore, the correlations in example \ref{ex3} exhibit Svetlichny steering from $A$ to $BC$ for $\eta>1/2$. 
\begin{remark}
Quantum correlations can also detect tripartite EPR-steering from Alice to Bob and Charlie even if Charlie's measurement
settings are compatible.
\end{remark}

\section{Conclusion}\label{Cnc}
We have introduced Svetlichny steering and Mermin steering to distinguish the presence of two types of genuine tripartite steering. 
We have derived the two inequalities which detect them by using the structure of the Svetlichny inequality and Mermin
inequality, which detect genuine nonlocality and standard nonlocality, respectively.
We find that genuine tripartite entanglement is certifiable from the correlations that violate one of these steering inequalities 
in a semi-device-independent way as well. That is, when qubit dimension is assumed for two parties,
the correlations which exhibit Svetlichny steering or Mermin steering also detect genuine tripartite entanglement (where measurements
on all subsystems are uncharacterized).
%We have observed that Svetlichny and Mermin steerability of statistics which 
%have LHV model originate from incompatible measurements that give rise 
%to the Svetlichny nonlocality and the GHZ paradox, respectively. 
We demonstrate that quantum correlations can also detect 
tripartite EPR-steering from Alice to Bob and Charlie even if
Charlie's measurement settings are compatible. 
%It would be interesting to investigate whether one can always 
%find a quantum state to demonstrate tripartite EPR-steering by using any pairwise
%noncommuting set of arbitrary number of POVMs. 
The two tripartite steering inequalities which detect Svetlichny steering and Mermin steering can be generalized to arbitrary number of systems by using the 
structure of $n$-partite Svetlichny inequality \cite{multiSI,multiSI1} and Mermin-Ardehali-Belinskii-Klyshko (MABK)
inequality \cite{mermin,multimermin,multimermin1,multimermin2,multimermin3}, respectively. 

\section*{Acknowledgement} 
I would like to thank Dr. J.-D. Bancal
and Prof. H. M. Wiseman for the helpful conversations on steering.
I wish to thank Dr. K. P. Yogendran and Prof. Suresh Govindarajan for useful comments.
I am thankful to Dr. Manik Banik for fruitful discussions. 
I am very thankful to an anonymous referee for 
the helpful comments and suggestions to improve the quality of the manuscript.
This work was partially funded by
Department of Science and Technology, via INSPIRE Project
No. PHY1415305DSTXPRAN, hosted at IIT Madras.
\bibliographystyle{apsrev4-1}
\bibliography{SDEW}
\end{document}